\documentclass[12pt, reqno]{amsart}

\makeatletter
\g@addto@macro{\endabstract}{\@setabstract}
\makeatother

\usepackage{graphics, stackrel}
\usepackage{amsmath, amssymb, amsthm}
\usepackage{graphicx}
\usepackage{verbatim}
\usepackage{amsfonts}
\usepackage{filecontents}
\usepackage{natbib}
\usepackage{fancyvrb}
\usepackage{color}

\usepackage[citecolor=blue, colorlinks=true, linkcolor=blue]{hyperref}

\usepackage{mathrsfs}
\usepackage{bbm}

\usepackage[left=1.25in, right=1.25in, top=1.0in, bottom=1.15in, includehead, includefoot]{geometry}

\usepackage{multirow}

\renewcommand{\leq}{\leqslant}
\renewcommand{\geq}{\geqslant}

\newcommand{\iidsim}{\stackrel {\textrm{ {\sc iid }}} {\sim} }

\newcommand*\diff{\mathop{}\!\mathrm{d}}
\renewcommand{\epsilon}{\varepsilon}

\newcommand{\bB}{\mathscr B}

\newcommand{\gG}{\mathcal G}
\newcommand{\vV}{\mathcal V}
\newcommand{\hH}{\mathcal H}

\newcommand{\XX}{\mathsf X}
\renewcommand{\AA}{\mathsf A}
\newcommand{\FF}{\mathsf F}
\newcommand{\DD}{\mathsf D}
\newcommand{\ZZ}{\mathsf Z}
\newcommand{\YY}{\mathsf Y}

\newcommand{\GG}{\mathbbm G}
\newcommand{\HH}{\mathbbm H}
\newcommand{\VV}{\mathbbm V}

\newcommand{\RR}{\mathbbm R}

\newcommand{\EE}{\mathbbm E \,}

\theoremstyle{plain}
\newtheorem{theorem}{Theorem}[section]

\newtheorem{lemma}[theorem]{Lemma}
\newtheorem{proposition}[theorem]{Proposition}

\theoremstyle{definition}

\newtheorem{assumption}{Assumption}[section]

\begin{document}

\title{}

\date{\today}

\begin{center}
	\Large 
    Dynamic optimal choice when rewards are unbounded below\footnote{We thank 
		Takashi Kamihigashi and Yiannis Vailakis 
		for valuable feedback and suggestions, as well as audience members
		at the Econometric Society meeting in Auckland in 2018 and the 2nd
		Conference on Structural Dynamic Models in Copenhagen.
		Financial support from ARC Discovery Grant
		DP120100321 is gratefully acknowledged. \\ 
		\emph{Email addresses:} \texttt{qingyin.ma@cueb.edu.cn}, \texttt{john.stachurski@anu.edu.au} }

	\bigskip
	\normalsize
	Qingyin Ma\textsuperscript{a} and John Stachurski\textsuperscript{b} \par \bigskip
	
	\textsuperscript{a}ISEM, Capital University of Economics and Business  \\
	\textsuperscript{b}Research School of Economics, Australian National University \bigskip
	
	\today
\end{center}

\begin{abstract} 
    We propose a new approach to solving dynamic decision problems with rewards that
    are unbounded below.  The approach involves transforming the Bellman
    equation in order to convert an unbounded problem into a bounded one.  The
    major advantage is that, when the conditions stated below
    are satisfied, the transformed problem can be solved by iterating with a 
    contraction mapping.  While the method is not universal, we show by example
    that many common decision problems do satisfy our conditions.
	
	\vspace{1em}
\end{abstract}


\section{Introduction}

Reward functions that are unbounded below have long been a stumbling block
for recursive solution methods, due to a failure of the standard contraction
mapping arguments first developed by \cite{blackwell1965discounted}. 
At the same time, such specifications are popular in economics and finance,
due to their convenience and well-established properties. This issue is
more than esoteric, since the Bellman equation for such problems can have
multiple solutions that confound the search for optima.
Computation of solutions, already
challenging when the state space is large, becomes even more so when
rewards are unbounded.


Here we propose a new approach to handling problems with values that are
unbounded below.  Instead of creating a new optimality theory, our approach
proceeds by transforming the Bellman equation to convert these unbounded
problems into bounded ones.  The main advantage of this approach is that, when
the conditions stated below are satisfied, the transformed problem can be
solved using standard methods based around contraction mappings. The technical
contribution of our paper lies obtaining suitable conditions and providing a
proof that the solution to the transformed problem is equal to the solution to
the original one.    While the method is not universal, we show by example
that many well-known decision problems do satisfy our conditions.

Our work contributes to a substantial existing literature on dynamic choice
with unbounded rewards.  The best known approach to such problems is the
weighted supremum norm method, originally developed by
\cite{wessels1977markov} and connected to economic modeling by
\cite{boyd1990recursive}.  This approach has been successful in treating many
maximization problems where rewards are unbounded above.   Unfortunately, as
noted by many authors, this same approach typically fails when rewards are
unbounded below.\footnote{See, for example, the discussions in
    \cite{le2005recursive} or \cite{jaskiewicz2011discounted}.  
    \cite{alvarez1998dynamic} find some success handling certain
problems that are unbounded below using weighted supremum norm methods,
although they require a form of homogeneity that fails to hold in the
applications we consider. \cite{bauerle2018stochastic} extend the weighted
supremum norm technique to risk sensitive preferences in a setting where
utility is bounded below.}

This failure was a major motivation behind the development of the local
contraction approach to dynamic programming, due to
\cite{rincon2003existence}, \cite{martins2010existence} and, for the
stochastic case \cite{matkowski2011discounted}.  This local contraction
method, which requires contractions on successively larger subsets of the
state space, is ingenious and elegant but also relatively technical, which might be the
cause of slow uptake on the part of applied economists.  A second disadvantage
in terms of applications is that the convergence results for value function
iteration are not as sharp as with traditional dynamic programming.

Another valuable contribution is \cite{jaskiewicz2011discounted}, which
explicitly admits problems with rewards that are unbounded below.  In this
setting, they show that the value function of a Markov decision process is a
solution to the Bellman equation.  We strengthen their results by adding a
uniqueness result and proving that value function iteration leads to an
optimal policy.  Both of these results are significant from an applied and
computational perspective.  Like \cite{jaskiewicz2011discounted}, we combine
our methodology with the weighted supremum norm approach, so that we can
handle problems that are both unbounded above and unbounded below.

Many other researchers have used transformations of the Bellman
equation, including \cite{rust1987optimal},
\cite{jovanovic1982selection}, \cite{bertsekas2017dynamic},
\cite{ma2018dynamic} and \cite{abbring2018very}.  These transformations are typically
aimed at improving economic intuition, estimation properties or computational efficiency.
The present paper is, to the best of our knowledge, the first to consider
transformations of the Bellman equation designed to solving dynamic
programming problems with unbounded rewards.

The rest of our paper is structured as follows.  Section~\ref{s:ea} starts the
exposition with typical examples. Section~\ref{s:gf} presents theory and
Section~\ref{s:app} gives additional applications.  Most proofs are deferred to
the appendix.

\section{Example Applications}

\label{s:ea}

We first illustrate the methodology for converting unbounded
problems to bounded ones in some common settings.

\subsection{Application 1: Optimal Savings}

\label{ss:os}

Consider an optimal savings problem where a borrowing constrained agent seeks
to solve
\begin{equation*}
	\sup \, \EE \sum_{t = 0}^\infty \beta^t u(c_t)
\end{equation*}
subject to the constraints
\begin{equation}
	\label{eq:bwc}
	0 \leq c_t \leq w_t,
	\quad  
	w_{t+1} = R (w_t - c_t) + y_{t+1}
	\quad \text{and} \quad
	(w_0,y_0) \; \text{given}.
\end{equation}
Here $\beta \in (0,1)$ is the discount factor, $c_t$, $w_t$ and $y_t$ are respectively consumption, wealth and non-financial income at time $t$, $R$ is the rate of return on financial income,\footnote{The timing associated with the wealth constraint in \eqref{eq:bwc} is such
	that $y_{t+1}$ is excluded from the time $t$ information set, as in, say
	\cite{benhabib2015wealth}.  One can modify the second constraint in
	\eqref{eq:bwc} to an alternative timing such as $w_{t+1} = R (w_t - c_t +
	y_t)$ and the arguments below still go through after suitable
	modifications.  An application along these lines is given in
	Section~\ref{ss:are}.}
 and $u$ is the CRRA utility function defined by
\begin{equation}
\label{eq:crra}
    u(c) = \frac{c^{1-\gamma} - 1}{1 - \gamma} \; \text{ with } \, \gamma > 1.
\end{equation}
%
We are focusing on the case $\gamma > 1$ because it is 
the most empirically relevant and, at the same time, the most challenging for
dynamic programming.

Assume that $\{y_t\}$ is a Markov process with state space $\YY \subset \RR_+$
and stochastic kernel $P$ satisfying\footnote{Here $P(y, \cdot \,)$ can be
interpreted as the transition probability. In particular, $P(y,A)$ represents
the probability of transitioning from $y$ to set $A$ in one step. See
Section~\ref{ss:t} for formal definition.}
\begin{equation}
	\label{eq:up}
	\bar u := \inf_{y \in \YY} \int u(y') P(y, \diff y') > - \infty.
\end{equation}
Condition \eqref{eq:up} holds if, say, 
\begin{itemize}
	\item $\{y_t\}$ is a finite state Markov chain taking positive values (see, e.g., \cite{Acikgoz2018} and \cite{cao2018recursive}), or
	\item $\{y_t\}$ is {\sc iid} and $\EE u(y_t) > -\infty$ (see, e.g., \cite{benhabib2015wealth}), or
	\item $\{y_t\}$ is a Markov switching process, say, $y_t = \mu_t + \sigma_t \epsilon_t$, where $\{\epsilon_t\} \iidsim N(0,1)$, while $\{\mu_t\}$ and $\{\sigma_t\}$ are positive and driven by finite state Markov chains (see, e.g., \cite{heathcote2010macroeconomic} and \cite{kaplan2010much}).
\end{itemize}
The Bellman equation of this problem is
\begin{equation}
	v(w, y) = \sup_{0 \leq c \leq w}
	\left\{
	u(c) + \beta \int v(R (w - c) + y', y') P(y,\diff y')
	\right\},
\end{equation}
where $w \in \RR_+$ and $y \in \YY$. Since $c_t \leq w_t$, it is clear that the value function is unbounded below.
Put differently, if $v$ is a candidate value function, then even if $v$ is
bounded, its image
\begin{equation}
	\label{eq:beos}
	Tv(w, y) = 
	\sup_{0 \leq c \leq w}
	\left\{
	u(c) + \beta \int v(R (w - c) + y', y') P(y, \diff y')
	\right\}
\end{equation}
under the Bellman operator is dominated by $u(w)$ plus some finite constant,
and hence $v(w, y) \to -\infty$ as $w \to 0$ for any $y \in \YY$.

Consider, however, the following transformation. Let $s := w - c$ and
\begin{equation}
	\label{eq:gfv}
	g(y,s) := \beta \int v(R s + y', y') P(y, \diff y')
\end{equation}
so that
\begin{equation}
	\label{eq:nv}
	v(w, y) = 
	\sup_{0 \leq s \leq w}
	\left\{
	u(w - s) + g(y,s)
	\right\}.
\end{equation}
We can eliminate the function $v$ from \eqref{eq:nv} 
by using the definition of $g$.  The first step is to evaluate $v$ in
\eqref{eq:nv} at $(R s + y', y')$, which gives 
\begin{equation*}
	v(Rs + y', y') = 
	\sup_{0 \leq s' \leq Rs + y'}
	\left\{
	u(Rs + y' - s') + g(y', s')
	\right\}.
\end{equation*}
Now we take expectations on both sides of the last equality and multiply
by $\beta$ to get
\begin{equation}
	\label{eq:ng}
	g(y,s) = 
	\beta \int
	\sup_{0 \leq s' \leq Rs + y'}
	\left\{
	u(Rs + y' - s') + g(y',s')
	\right\} P(y, \diff y').
\end{equation}
This is a functional equation in $g$. We now introduce a modified Bellman operator $S$ such that any solution $g$ of \eqref{eq:ng} is a fixed point of $S$:
\begin{equation}
	\label{eq:sg}
	Sg(y, s) = 
	\beta \int
	\sup_{0 \leq s' \leq Rs + y'}
	\left\{
	u(Rs + y' - s') + g(y',s')
	\right\} P(y, \diff y').
\end{equation}
Let $\gG$ be the set of bounded measurable functions on $\YY \times \RR_+$.
We claim that $S$ maps $\gG$ into itself and, moreover, is a contraction of
modulus $\beta$ with respect to the supremum norm.

To see that this is so, pick any $g \in \gG$. Then $Sg$ is bounded above, since $\gamma > 1$ implies
\begin{equation*}
	Sg(y,s) \leq \beta (\sup_{c \geq 0} u(c) + \|g\|) \leq \beta \|g\|,
\end{equation*}
where $\| \cdot \|$ is the supremum norm.  More importantly, $Sg$ is bounded below. Indeed,
\begin{align*}
	Sg(y,s) 
	& \geq 
	\beta \int
	\sup_{0 \leq s' \leq Rs + y'}
	\left\{
	u(Rs + y' - s') - \| g \|
	\right\} P(y, \diff y')
	\\
	& =
	\beta \int
	\left\{
	u(Rs + y') - \| g \|
	\right\} P(y, \diff y')
	 \geq 
	\beta \int u(y') P(y, \diff y') - \beta \| g \|
	\geq 
	\beta \, \bar u - \beta \| g \|.
\end{align*}
Finally, $S$ is obviously a contraction mapping, since, for any $g, h \in \GG$, we have
\begin{multline*}
	\left|
	\sup_{s'}
	\left\{
	u(Rs + y' - s') + g(y',s')
	\right\} 
	-
	\sup_{s'}
	\left\{
	u(Rs + y' - s') + h(y', s')
	\right\} 
	\right| 
	\\
	\leq \sup_{s'} | g(y', s') - h(y', s') |
\end{multline*}
and hence
\begin{equation*}
	|Sg(y,s) - Sh(y,s) |
	\leq
	\beta 
	\int \sup_{0 \leq s' \leq Rs + y'}
	| g(y', s') - h(y', s') | P(y, \diff y')
	 \leq \beta
	\| g - h \|.
\end{equation*}
Taking the supremum over all $(y,s) \in \YY \times \RR_+$ yields
\begin{equation*}
	\| Sg - Sh \| \leq \beta \| g - h \|.
\end{equation*}
%
We have now shown that $S$ is a
contractive self-map on $\gG$. Most significant here is that $\gG$ is a space 
of bounded functions. By Banach's contraction mapping theorem,
$S$ has a unique fixed point $g^*$ in $\gG$.  Presumably, we can insert $g^*$
into the right hand side of the ``Bellman equation'' \eqref{eq:ng}, compute
the maximizer at each state and obtain the optimal savings policy.  If a
version of Bellman's principle of optimality applies to this modified Bellman
equation, we also know that policies obtained in this way exactly coincide
with optimal policies, so, if all of these conjectures are correct, we have a
complete characterization of optimality.

A significant amount of theory must be put in place to make the proceeding
arguments work. In particular, the conjectures discussed immediately above
regarding the validity of Bellman's principle of optimality vis-a-vis the
modified Bellman equation are nontrivial, since the transformation in
\eqref{eq:gfv} that maps $v$ to $g$ is not bijective.  As a result, some
careful analysis is required before we can make firm conclusions regarding
optimality.  This is the task of Section~\ref{s:gf}.

A final comment on this application is that, for this particular problem, we
can also use Euler equation methods, which circumvent some of the issues
associated with unbounded rewards (see, e.g., \cite{li2014solving}).
However, these methods are not applicable in many other settings, due to
factors such as 
existence of discrete choices.  The next two applications
illustrate this point.


\subsection{Application 2: Job Search}

\label{ss:jsp}

As in \cite{mccall1970economics}, an unemployed worker can either accept
current job offer $w_t = z_t + \xi_t$ and work at that wage forever or choose
an outside option (e.g., irregular work in the informal
sector) yielding $c_t = z_t + \zeta_t$ and continue to the next period.
Here $z_t$ is a persistent component, while $\xi_t$ and $\zeta_t$
are transient components.  We assume that $\{\xi_t\}$ and $\{\zeta_t\}$ are
{\sc iid} and lognormal, and 
\begin{equation}
	\label{eq:zp}
	\ln z_{t+1} = \rho \ln z_t + \sigma \epsilon_{t+1},
	\quad \{\epsilon_t\} \iidsim N(0, 1).
\end{equation}
The worker's value function satisfies the Bellman equation
\begin{equation}
	\label{eq:mcvf}
	v(w, c, z) = \max 
	\left\{
	\frac{u(w)}{1-\beta},\;
	u(c) + \beta \EE_{z}
	\, v(w', c', z') 
	\right\}.
\end{equation}
Let $u$ be increasing, continuous, and unbounded below with $u(w) = - \infty$ as $w \to 0$. For now, let $u$ be bounded above. Moreover, we assume that
\begin{equation}
	\label{eq:up2}
	\text{either } \; \inf_{z>0} \, \EE_{z} u(w') > - \infty
	\quad \text{or} \quad
	\inf_{z>0} \, \EE_{z} u(c') > - \infty.
\end{equation}
Condition~\eqref{eq:up2} is satisfied if $u$ is CRRA, say,
since then $\EE u(\xi_t)$ and $\EE u(\zeta_t)$ are finite. Note that $v(w,c,z)$ is unbounded
below since utility can be arbitrarily close to $-\infty$.

To shift to a bounded problem, we can proceed in a similar vein to our
manipulation of the Bellman equation in the optimal savings case. First we set
\begin{equation*}
	g(z) := \beta \EE_{z} \, v(w', c', z'),
\end{equation*}
so that \eqref{eq:mcvf} can be written as 
\begin{equation*}
	v(w, c, z) = \max 
	\left\{
	\frac{u(w)}{1-\beta},\;
	u(c) + g(z)
	\right\}.
\end{equation*}
Next we use the definition of $g$ to eliminate $v$ from this last expression,
which leads to the functional equation
\begin{equation}
	\label{eq:hfe}
	g(z) = \beta \EE_{z} 
	\max 
	\left\{
	\frac{u(w')}{1-\beta},\;
	u(c') + g(z')
	\right\}.
\end{equation}
The corresponding fixed point operator is
\begin{equation}
	\label{eq:shfe}
	Sg(z) = \beta \EE_{z} 
	\max 
	\left\{
	\frac{u(w')}{1-\beta},\;
	u(c') + g(z')
	\right\}.
\end{equation}
If $g$ is bounded above then clearly so is $Sg$.
Moreover, if $g$ is bounded below by some constant $M$, then, by Jensen's
inequality,  
\begin{align*}
	Sg(z) 
	& \geq 
	\beta 
	\max 
	\left\{
	\EE_{z}
	\frac{u(w')}{1-\beta}, \;
	\EE_{z} u(c') + M
	\right\}.
\end{align*}
Condition~\eqref{eq:up2} then implies that $Sg$ is also bounded below.

An argument similar to the one adopted above for the optimal savings model
proves that $S$ is a contraction mapping with respect to the supremum
norm on a space of bounded functions (Section~\ref{s:gf} gives details).
Thus, we can proceed down essentially the same path we used for the optimal
savings problem, with the same caveat that the modified Bellman operator $S$
and the original Bellman operator need to have the same connection to
optimality, and all computational issues need to be clarified.

\subsection{Application 3: Optimal Default}

\label{ss:are}

Consider an infinite horizon optimal savings problem with default, in the
spirit of \cite{arellano2008default} and a large related
literature.\footnote{Recent examples include 
\cite{AGAM2019} and \cite{AAHW2019}.}
    A country
with current assets $w_t$ chooses between continuing to participate in
international financial markets and default.  Output 
\begin{equation*}
	y_t = y(z_t, \xi_t)
\end{equation*}
is a function of a persistent component $\{z_t\}$ and an innovation
$\{\xi_t\}$.  The persistent component is a Markov process such as
the one in \eqref{eq:zp} and the transient component $\{\xi_t\}$ is {\sc
	iid}.  To simplify the exposition, we assume that default leads to permanent
exclusion from financial markets, with lifetime value
\begin{equation*}
	v^d(y, z) = \EE \sum_{t = 0}^\infty \beta^t u(y_t).
\end{equation*}
Notice that $v^d$ satisfies the functional equation
\begin{equation*}
    v^d(y,z) = u(y) + \beta \EE_{z} v^d (y',z').
\end{equation*}
The value of continued participation in financial markets is 
\begin{equation*}
	v^c(w, y, z)
	= \sup_{-b \leq w' \leq R(w + y)}
	\left\{
	u(w + y - w'/R) + \beta \EE_{z} \, v(w', y', z')
	\right\},
\end{equation*}
where $b > 0$ is a constant borrowing
constraint and $v$ is the value function satisfying
\begin{equation*}
	v(w, y, z)
	= \max
	\left\{
	v^d(y, z)
	,\, 
	v^c(w, y, z)
	\right\}.
\end{equation*}
The utility function $u$ has the same properties as Section~\ref{ss:jsp}. It is easy to see that $v$ is unbounded below since $u$ can be arbitrarily close to $-\infty$. However, we can convert this into a bounded problem, as the following analysis shows. 

Let $i$ be a discrete choice variable taking values in $\{0,1\}$, with $0$ indicating default and $1$ indicating continued participation. We define
\begin{center}
	$g(z,w',i) := 
	\begin{cases}
	\beta \EE_{z} \, v^d(y', z') & \text{if } \; i=0  \\
	\beta \EE_{z} \, v(w', y', z') & \text{if } \; i=1
	\end{cases}$
\end{center}
%
so that for $-b \leq w' \leq R(w + y)$, we have
\begin{equation*}
	v(w, y, z)
	= \max
	\left\{
	u(y) + g(z,w',0) 
	,\,
	\sup_{w'}
	\left\{
	u(w + y - w'/R) + g(z, w',1)
	\right\}
	\right\}.
\end{equation*}
Eliminating the value function $v$ yields
\begin{equation*}
    g(z,w',0) = \beta \EE_{z} \{u(y') + g (z',w',0)\}
    \quad \text{and}
\end{equation*}
\begin{equation*}
	g(z, w',1) 
	= \beta \EE_{z}
	\max
	\left\{
	u(y') + g(z',w',0) 
	, \,
	\sup_{w''}
	\left\{
	u(w' + y' - w''/R) + g(z',w'',1)
	\right\}
	\right\}, 
\end{equation*}
%
where $-b \leq w'' \leq R(w' + y')$. We can then define the fixed point operator $S$ corresponding to these functional equations.

If $g$ is bounded above by some constant $K$, then $Sg \leq \sup_c u(c) +
K$.  More importantly, if $g$ is bounded below by some constant
$M$, we obtain
\begin{equation*}
    Sg(z,w',0) \geq \beta \EE_{z} u(y') + \beta M 
    \quad \text{and} 
\end{equation*}
\begin{align*}
	Sg(z, w',1)
	& \geq
	\beta \EE_{z}
	\max
	\left\{
	u(y') + M 
	,\,
	u(w' + y' + b/R) + M
	\right\}
	\\
	& =
	\beta \EE_{z}
	\max
	\left\{
	u(y') 
	,\,
	u(w' + y' + b/R)
	\right\}
	+ \beta M.
\end{align*}
Hence, $Sg$ is bounded below by a finite constant if
\begin{equation}
	\label{eq:odbb}
	\inf_{z} \EE_{z} u(y') > -\infty.
\end{equation}
For example, \eqref{eq:odbb} holds if $y_t = z_t + \xi_t$ where $\{z_t\}$ is positive and $\EE u(\xi_t) > -\infty$. 
An argument similar to the one in
Section~\ref{ss:os} now proves that $S$ is a contraction with respect
to the supremum norm (Section~\ref{s:gf} gives details).

\section{General Formulation}

\label{s:gf}

The preceding section showed how some unbounded problems can be converted to
bounded problems by modifying the Bellman equation.
The next step is to confirm the validity of such a
modification in terms of the connection between the modified Bellman equation
and optimal policies.
We do this in a generic dynamic programming setting that
contains the applications given above.

\subsection{Theory}
\label{ss:t}

For a given set $E$, let $\bB(E)$ be the Borel subsets of $E$.  For our
purpose, a dynamic program consists of
\begin{itemize}
	\item a nonempty set $\XX$ called the \textit{state space},
	\item a nonempty set $\AA$ called the \textit{action space},
	\item a nonempty correspondence $\Gamma$ from $\XX$ to $\AA$ called the \textit{feasible correspondence}, along with the associated set of \textit{state action pairs}
	\begin{equation*}
	    \DD := \{ (x,a) \in \XX \times \AA : a \in \Gamma (x) \},
	\end{equation*}
	%
	\item a measurable map $r: \DD \to \RR \cup \{-\infty\}$ called the \textit{reward function},
    \item a constant $\beta \in (0,1)$ called the \textit{discount factor}, and 
    \item a \textit{stochastic kernel} $Q$ governing the evolution of states.\footnote{Here a 
    	\textit{stochastic kernel} corresponding to our controlled Markov process $\{(x_t,a_t)\}$ is a mapping $Q: \DD \times \bB(\XX) \to [0,1]$ such that (i) for each $(x,a) \in \DD$, $A \mapsto Q(x,a, A)$ is a probability measure on $\bB(\XX)$, and (ii) for each $A \in \bB(\XX)$, $(x,a) \mapsto Q(x,a,A)$ is a measurable function.}
\end{itemize}
Each period, an agent observes a state $x_t \in \XX$ and responds with an
action $a_t \in \Gamma(x_t) \subset \AA$. The agent then obtains a reward
$r(x_t, a_t)$, moves to the next period with a new state $x_{t+1}$, and
repeats the process by choosing $a_{t+1}$ and so on. The state process updates
according to $x_{t+1} \sim Q (x_t, a_t, \cdot \,)$.

Let $\Sigma$ denote the set of \textit{feasible policies}, which we assume to
be nonempty and define as all measurable maps $\sigma: \XX \to \AA$ satisfying
$\sigma (x) \in \Gamma(x)$ for all $x \in \XX$.  Given any policy $\sigma \in
\Sigma$ and initial state $x_0 = x \in \XX$, the \textit{$\sigma$-value
function} $v_\sigma$ is defined by
%
\begin{equation*}
	v_\sigma (x) = \sum_{t=0}^\infty \beta^t \EE_x r(x_t, \sigma(x_t)).
\end{equation*}
We understand $v_\sigma (x)$ as the lifetime value of following policy $\sigma$ now and forever, starting from current state $x$.  

The \textit{value function} associated with this dynamic program is defined at each $x \in \XX$ by 
\begin{equation}
	\label{eq:szvf}
	v^*(x) = \sup_{\sigma \in \Sigma} v_{\sigma} (x).
\end{equation}
A feasible policy $\sigma^*$ is called \textit{optimal} if $v_{\sigma^*} =
v^*$ on $\XX$. The objective of the agent is to find an optimal policy that
attains the maximum lifetime value. 

To handle rewards that are unbounded above as well as below, we introduce a
weighting function $\kappa$, which is a measurable function mapping $\XX$ to
$[1, \infty)$. Let $\gG$ be the set of measurable functions $g: \DD \to \RR$
such that $g$ is bounded below and 
\begin{equation}
    \|g\|_\kappa := \sup_{(x,a) \in \FF} \, \frac{|g(x,a)|}{\kappa (x)} < \infty.
\end{equation}
The pair $(\gG, \|\cdot\|_\kappa)$ is a Banach space (see, e.g., \cite{bertsekas2013abstract}).
Moreover, at each $x \in \XX$ and $(x,a) \in \DD$, we define 
\begin{equation}
\label{eq:rbar}
    \bar r (x) := \sup_{a \in \Gamma (x)} r (x,a)
    \quad \text{and} \quad
    \ell (x,a) := \EE_{x,a} \bar r (x').
\end{equation}

\begin{assumption}
	\label{a:ws}
	There exist constants $d \in  \RR_+$ and $\alpha \in (0, 1 / \beta)$ such that $\bar r (x) \leq d \kappa (x)$ and
	$\EE_{x,a} \, \kappa (x') \leq \alpha \kappa(x)$ for all $(x,a) \in \DD$.
\end{assumption}

Assumption~\ref{a:ws} relaxes the standard weighted supremum norm assumptions (see, e.g., \cite{wessels1977markov} or \cite{bertsekas2013abstract}), in the sense that the reward function is allowed to be unbounded from below.

Next, we define  $S$ on $\gG$ as 
\begin{equation}
	\label{eq:rfba}
	S g(x,a) := \beta \EE_{x,a} \sup_{a' \in \Gamma(x')}
	\left\{ 
	    r \left( x', a' \right) + g(x',a')
	\right\}.
\end{equation}
Given $g \in \gG$, a feasible policy $\sigma$ is called \emph{$g$-greedy} if
\begin{equation} 
	\label{eq:greedy}
	r(x, \sigma(x)) + g(x, \sigma(x))
	= \sup_{a \in \Gamma(x)} \{ r(x,a) + g(x,a) \}
	\quad \text{for all $x \in \XX$.}
\end{equation}
Although the reward function is potentially unbounded below, the dynamic
program can be solved by the operator $S$, as the following theorem shows.

\begin{theorem}
	\label{t:cs}
	If Assumption~\ref{a:ws} holds and $\ell$ is bounded below, then 
	\begin{enumerate}
		\item $S \gG \subset \gG$ and $S$ is a contraction mapping on $(\gG, \| \cdot \|_\kappa)$.
		\item $S$ admits a unique fixed point $g^*$ in $\gG$. 
		\item $S^k g$ converges to $g^*$ at rate $O((\alpha \beta)^k)$ under $\| \cdot \|_{\kappa}$.
		\item If there exists a closed subset $\GG$ of $\gG$ such that $S \GG
            \subset \GG$ and a $g$-greedy policy exists for each $g \in \GG$,
            then, in addition,
		\begin{enumerate}
			\item $g^*$ is an element of $\GG$ and satisfies 
			\begin{equation*}
                \qquad
				g^*(x,a) = \beta \EE_{x,a} v^*(x') 
				\quad \text{and} \quad
				v^* (x) = \max_{a \in \Gamma(x)} \left\{ r(x,a) + g^*(x,a) \right\}.
			\end{equation*}
			\item At least one optimal policy exists.
			\item A feasible policy is optimal if and only if it is $g^*$-greedy.
		\end{enumerate}
	\end{enumerate}
\end{theorem}

\subsection{Sufficient Conditions}
\label{ss:suff}

Consider a dynamic programming problem 
\begin{equation}
    \label{eq:sc_obj}
    \max \, \EE \sum_{t=0}^\infty \beta^t r(w_t, s_t)
\end{equation}
subject to
\begin{equation}
	\label{eq:ubdd_st}
	0 \leq s_t \leq w_t, \quad 
	w_{t+1} = f(s_t, \eta_{t+1}), \quad
	\eta_t = h (z_t, \epsilon_t)
	\quad \text{and} \quad 
	(w_0, z_0) \; \text{given}.
\end{equation}
Here $z$ and $\epsilon$ correspond respectively to a Markov process $\{z_t\}$ on $\ZZ$ and an {\sc iid} process $\{\epsilon_t\}$, $f$ and $h$ are nonnegative continuous functions, and $f$ is increasing in $s$. Furthermore, $\mathsf{Z}$ and the range space of $\{\eta_t\}$ are Borel subsets of finite-dimensional Euclidean spaces, and the stochastic kernel $P$ corresponding to $\{z_t\}$ is Feller.\footnote{In other words,
    $z \mapsto \int h(z') P(z,\diff z')$ is bounded and continuous whenever $h$ is.}

This problem can be placed in our framework by setting
\begin{equation*}
    x := (w,z), \quad
    a := s, \quad
    \XX := \RR_+ \times \ZZ, \quad 
    \AA := \RR_+, \quad
    \Gamma(x) := [0, w] 
\end{equation*}
\begin{equation*}
    \text{and} \quad
    \DD := \left\{ (w,z,s) \in \RR_+ \times \ZZ \times \RR_+ : 0 \leq s \leq w \right\}.
\end{equation*}

Suppose that the reward function $r: \DD \to \RR \cup \{ -\infty \}$ is increasing in $w$ and decreasing in $s$, $r$ is continuous on the interior of $\DD$ and, if $r$ is bounded below, it is continuous. 

Recall $\kappa$ defined in Assumption~\ref{a:ws}. Let
\begin{equation*}
    \underline{\ell} (z) := \EE_{z} r (f(0,\eta'), 0) 
    \quad \text{and} \quad
    \kappa_e (z, s) := \EE_{z,s} \kappa(w',z').
\end{equation*}
Let $\GG$ be the set of functions $g$ in $\gG$ that is increasing in its last argument and continuous. 
Notice that, in the current setting, $S$ defined on $\GG$ is given by
\begin{equation*}
    Sg (z,s) = \beta \EE_{z, s} \max_{s' \in [0,w']} \left\{
        r(w',s') + g(z',s')
    \right\}.
\end{equation*}
Theorem~\ref{t:cs} is applicable in the current setting, as the following result illustrates.

\begin{proposition}
	\label{pr:suff}
	If Assumption~\ref{a:ws} holds for some continuous functions $\kappa$ and $\kappa_e$, and $\underline{\ell}$ is continuous and bounded below,   
    then $S$ is a contraction mapping on $(
    \GG, \| \cdot \|_{\kappa} )$ and the conclusions of Theorem~\ref{t:cs}
    hold.
\end{proposition}

\section{Applications}
\label{s:app}

In this section, we complete the discussion of all applications in
Section~\ref{s:ea}. We also extend the optimality results of
\cite{benhabib2015wealth} by adding a persistent component to labor income and
returns.

\subsection{Optimal Savings (Continued)}
\label{ss:os_cont}

Recall the optimal savings problem of Section~\ref{ss:os}. This problem can be placed into the framework of Section~\ref{ss:suff} by letting
\begin{equation*}
    \eta = z := y, \quad
    r(w,s) := u(w - s), \quad
    f(s, \eta') := R s + \eta'
    \quad \text{and} \quad
    h (z, \epsilon) := z.
\end{equation*}
To establish the desired properties, it remains to verify the conditions of
Proposition~\ref{pr:suff}. Since we have shown that $S \gG \subset \gG$, where
$\gG$ is the set of bounded measurable functions on $\YY \times \RR_+$, we can
simply set $\kappa \equiv 1$ such that Assumption~\ref{a:ws} holds. In this
case, both $\kappa$ and $\kappa_e$ are continuous functions. Moreover, note
that 
\begin{equation*}
    \underline{\ell} (y) = \EE_{y} u(y') 
    = \int u(y') P(y, \diff y'),
\end{equation*}
which is bounded below by \eqref{eq:up}. As a result, all the conclusions of Theorem~\ref{t:cs} hold as long as $y \mapsto \int u(y') P(y, \diff y')$ is continuous. In particular, when this further condition holds, $S$ is a contraction mapping on $(\GG, \| \cdot\|)$ with unique fixed point $g^*$, and a feasible policy is optimal if and only if it is $g^*$-greedy. Here $\GG$ is the set of bounded continuous functions on $\YY \times \RR_+$ that is increasing in its last argument.

\subsection{Job Search (Continued)}
\label{ss:js_cont}

Recall the job search problem of Section~\ref{ss:jsp}. This problem fits into the framework of Section~\ref{ss:t} if we let the  $a$ be a discrete choice variable taking values in $\{0,1\}$, where $0$ denotes the decision to stop and $1$ represents the decision to continue,
\begin{equation*}
    x := (w,z,c), \;\; 
    \XX := (0, \infty)^3, \;\;
    \AA := \{0,1\}, \;\;
    \Gamma(x) := \{0,1\}, \;\;
    \DD := (0, \infty)^3 \times \{0,1\}
\end{equation*}
and the reward function $r(x,a)$ be
\begin{equation*}
    r(w,c,a) := \frac{u(w)}{1-\beta} 
    \; \text{ if } \; a = 0
    \quad \text{and} \quad
    r(w,c,a) := u(c) 
    \; \text{ if } \; a = 1.
\end{equation*}
We have shown that $S \gG \subset \gG$, where $\gG$ is the set of bounded measurable functions on $(0, \infty)$. Hence, Assumption~\ref{a:ws} holds with $\kappa \equiv 1$. Note that in this case, the function $\ell (x,a)$ reduces to
\begin{equation*}
    \ell (z) = \EE_{z} \max \left\{ u(w') / (1 - \beta), u(c') \right\}.
\end{equation*}
Then $\ell$ is bounded below by Jensen's inequality and \eqref{eq:up2}. Since in addition the action set is finite, a $g$-greedy policy always exists for all $g \in \gG$. 
Let $\GG := \gG$. The analysis above implies that all the conclusions of Theorem~\ref{t:cs} hold.

\subsection{Optimal Default (Continued)}
\label{ss:od_cont}

Recall the optimal default problem studied in Section~\ref{ss:are}. This setting is a special case of our framework. In particular, 
\begin{equation*}
    x := (w, y, z), \;\;\;
    a := (w', i), \;\;\;
    \XX := [-b, \infty) \times \YY \times \ZZ
    \;\;\; \text{and} \;\;\;
    \AA := [-b, \infty) \times \{ 0,1 \},
\end{equation*}
where $i$ is a discrete choice variable taking values in $\{0,1\}$, and $\YY$ and $\ZZ$ are respectively the range spaces of $\{y_t\}$ and $\{z_t\}$. The reward function $r$ reduces to
\begin{center}
	$r(w,y,w', i) := 
	\begin{cases}
	    u(y) & \text{if } \; i=0,  \\
	    u(w + y - w'/R) & \text{if } \; i=1.
	\end{cases}$
\end{center}
Since $S \gG \subset \gG$, where $\gG$ is the set of bounded measurable functions on $\ZZ \times [-b, \infty) \times \{0,1\}$, Assumption~\ref{a:ws} holds for $\kappa \equiv 1$. Moreover, $\ell$ satisfies
\begin{equation*}
    \ell (z, w') = \EE_{z} \max \left\{ u(y'), u \left( w' + y' + b/R \right) \right\}
    \geq \EE_z u(y'),
\end{equation*}
which is bounded below by \eqref{eq:odbb}. Let $\GG$ be the set of functions
in $\gG$ that is increasing in its second-to-last argument and continuous.
Through similar steps to the proof of Proposition~\ref{pr:suff}, one can show
that $S \GG \subset \GG$ and a $g$-greedy policy exists for all $g \in \GG$.
As a result, all the conclusions of Theorem~\ref{t:cs} are true.

\subsection{Optimal Savings with Capital Income Risk}
\label{ss:ops_cir}

Consider an optimal savings problem with capital income risk (see, e.g., \cite{benhabib2015wealth}). The setting is similar to that of Section~\ref{ss:os}, except that the rate of return to wealth is stochastic. In particular, the constraint \eqref{eq:bwc} now becomes
\begin{equation*}
    0 \leq s_t \leq w_t, \quad
    w_{t+1} = R_{t+1} s_t + y_{t+1}
    \quad \text{and} \quad
    (w_0, z_0) \text{ given}.
\end{equation*}
where $w_t$ is wealth, $s_t$ is the amount of saving, while $R_t$ and $\{y_t\}$ are respectively the rate of return to wealth and the non-financial income that satisfy
\begin{equation*}
    R_t = h_R (z_t, \xi_t) 
    \quad \text{and} \quad
    y_t = h_y (z_t, \zeta_t).
\end{equation*}
Here $\{z_t\}$ is a finite state Markov chain, and $\{\xi_t\}$ and $\{\eta_t\}$ are {\sc iid} innovation processes. The importance of these features for wealth dynamics is highlighted in \cite{fagereng2016heterogeneityNBER} and \cite{hubmer2018comprehensive}, among others.

This problem fits into the framework of Section~\ref{ss:suff} by setting 
\begin{equation*}
    \eta := (R, y), \quad
    \epsilon_t := (\xi_t, \zeta_t), \quad
    r(w,a) := u(w-s)
    \quad \text{and} \quad
    f(s, \eta') := Rs + y'.
\end{equation*}
%
%
In this case, $\underline{\ell} (z) = \EE_{z} u(y')$ and 
\begin{equation*}
	\label{eq:opsav_I_be1}
	Sg(z, s) = \beta \EE_{z,s} \max_{s' \in [0,w']} 
	\left\{ u(w' - s') + g(z',s') \right\}.
\end{equation*}
Consider, for example, the CRRA utility in \eqref{eq:crra}. In this case, Assumption~\ref{a:ws} holds with $\kappa \equiv 1$, and $\GG$ reduces to the set of bounded continuous functions on $\ZZ \times \RR_+$ that is increasing in its last argument. The conclusions of Theorem~\ref{t:cs} hold if $z \mapsto \EE_{z} u(y')$ is continuous and bounded below.

\section{Appendix}
\label{s:appendix}


Let $\vV$ (resp., $\VV$) be the set of measurable functions $v: \XX \to \RR
\cup \{-\infty\}$ such that $(x,a) \mapsto \beta \EE_{x,a} v(x')$ is in $\gG$
(resp., $\GG$), and let $\hH$ (resp., $\HH$) be the set of measurable
functions $h:\DD \to \RR \cup \{ -\infty \}$ such that $h = r + g$ for some
$g$ in $\gG$ (resp., $\GG$). Next, we define the operators $W_0$, $W_1$ and
$M$ respectively on $\vV$, $\gG$ and $\hH$ as
\begin{equation*}
	W_0 v (x,a) := \beta \EE_{x,a} v(x'),
	\quad
	W_1 g (x,a) := r(x,a) + g(x,a),
\end{equation*}
\begin{equation*}
    \quad \text{and} \quad
    M h (x) := \sup_{a \in \Gamma(x)} h (x,a). 
\end{equation*}
Then $S$ in \eqref{eq:rfba} satisfies $S = W_0 M W_1$ on $\gG$.

\begin{proof}[Proof of Theorem~\ref{t:cs}]
    To see claim (1) holds, we first show that $S \gG \subset \gG$. Fix $g \in \gG$. By the definition of $\gG$, there is a lower bound $\underline{g} \in \RR$ such that $g \geq \underline{g}$. Then
    \begin{align*}
        S g (x,a) 
        &\geq \beta \EE_{x,a} \sup_{a' \in \Gamma(x')} \left\{
        r(x',a') +  \underline{g} \right\}   \\
        &= \beta  \left[ \EE_{x,a} \sup_{a' \in \Gamma(x')} r(x',a') + \underline{g} \right]
        = \beta \left[\EE_{x,a} \bar{r} (x') + \underline{g} \right] 
        = \beta \left[ \ell (x,a) + \underline{g} \right].
    \end{align*}
    Since by assumption $\ell$ is bounded below, so is $Sg$. Moreover, by Assumption~\ref{a:ws},
    \begin{align*}
        S g (x,a) &\leq \beta \EE_{x,a} \left\{ \bar r (x')
        + \sup_{a' \in  \Gamma(x')} g(x',a') \right\}    \\
        &\leq \beta \EE_{x,a} \left\{ (d +  \|g\|_\kappa) \kappa (x') \right\}    
        \leq \alpha \beta (d + \|g\|_\kappa) \kappa (x) 
    \end{align*}
    for all $(x,a ) \in \DD$. Hence, $Sg / \kappa$ is bounded above. Since in addition $Sg$ is bounded below and $\kappa \geq 1$, we have $\|Sg\|_\kappa < \infty$. We have now shown that $Sg \in \gG$.
    
    Next, we show that $S$ is a contraction mapping on $(\gG, \| \cdot \|_\kappa)$. Fix $g_1, g_2 \in \gG$. Note that for all $(x,a) \in \DD$, we have
    \begin{align*}
        &|S g_1 (x,a) - S g_2 (x,a)|    \\
        &=\left| \beta \EE_{x,a} \sup_{a' \in  \Gamma(x')} \{ r(x',a') +  g_1 (x',a') \} - 
        \beta \EE_{x,a} \sup_{a' \in  \Gamma(x')} \{ r(x',a') +  g_2 (x',a') \}  \right|    \\
        & \leq \beta \EE_{x,a} 
        \left| \sup_{a' \in  \Gamma(x')} \{ r(x',a') + g_1 (x',a') \} - 
        \sup_{a' \in  \Gamma(x')} \{ r(x',a') + g_2 (x',a') \} \right|    \\
        & \leq \beta \EE_{x,a} \sup_{a' \in \Gamma(x')}  
        \left| g_1(x', a') - g_2 (x',a') \right| 
        \leq \beta \| g_1 - g_2\|_\kappa \EE_{x,a} \kappa(x')
        \leq \alpha \beta \|g_1 - g_2 \|_\kappa \kappa(x),
    \end{align*}
    where the last inequality follows from Assumption~\ref{a:ws}. Then we have $\|S g_1 - S g_2 \|_\kappa \leq \alpha \beta \| g_1 - g_2 \|_\kappa$. Since $\alpha \beta < 1$, $S$ is a contraction mapping on $(\gG, \| \cdot \|_\kappa)$ and claim~(1) is verified.
    
    Claims~(2)--(3) follow immediately from claim~(1) and the Banach contraction mapping theorem. Regarding claim~(4), since $\GG$ is a closed subset of $\gG$ and $S \GG \subset \GG$, $S$ is also a contraction mapping on $(\GG, \| \cdot \|_\kappa)$ and the unique fixed point $g^*$ of $S$ is indeed in $\GG$. Based on Proposition~2 of \cite{ma2018dynamic}, the Bellman operator $T := M W_1 W_0$ maps elements of $\VV$ into itself and has a unique fixed point $\bar v$ in $\VV$ that satisfies $\bar v = M W_1 g^*$ and $g^* = W_0 \bar v$. 
    
    To verify part~(a) of claim~(4), it remains to show that $\bar v = v^*$.
    For all $x_0 \in \XX$ and $\sigma \in \Sigma$, we have
    \begin{align}
    \label{eq:vbar}
    	\bar v(x_0) &\geq r(x_0, \sigma(x_0)) + \beta \EE_{x_0,\sigma(x_0)} \bar{v} (x_1)  \nonumber  \\
    	& \geq r(x_0, \sigma(x_0)) + \beta \EE_{x_0,\sigma(x_0)} 
    	\left\{ r(x_1, \sigma(x_1)) + \beta \EE_{x_1,\sigma(x_1)} \bar{v} (x_2) \right\}  \nonumber  \\
    	&= r(x_0, \sigma(x_0)) + \beta \EE_{x_0,\sigma(x_0)} 
    	r(x_1, \sigma(x_1)) + \beta^2 \EE_{x_0,\sigma(x_0)} \EE_{x_1,\sigma(x_1)} \bar{v} (x_2)  \nonumber  \\
    	&\geq \sum_{t=0}^T \beta^t \EE_{x_0,\sigma(x_0)} \cdots \EE_{x_{t-1},\sigma(x_{t-1})} r(x_t, \sigma(x_t))  
    	 + \beta^{T+1} \EE_{x_0,\sigma(x_0)} \cdots \EE_{x_{T},\sigma(x_{T})} \bar v (x_{T+1})  \nonumber  \\
    	&= \sum_{t=0}^T \beta^t \EE_{x_0} r(x_t, \sigma(x_t)) + 
    	\beta^{T} \EE_{x_0,\sigma(x_0)} \cdots \EE_{x_{T-1},\sigma(x_{T-1})} g^* (x_T, \sigma(x_T)).
    \end{align}
    Notice that, by Assumption~\ref{a:ws}, we have
    \begin{align*}
        & \left|\beta^{T} \EE_{x_0,\sigma(x_0)} \cdots \EE_{x_{T-1},\sigma(x_{T-1})} g^* (x_T, \sigma(x_T)) \right|    \\
        & \leq \beta^{T} \EE_{x_0,\sigma(x_0)} \cdots \EE_{x_{T-1},\sigma(x_{T-1})} \left| g^* (x_T, \sigma(x_T)) \right|    \\
        & \leq \beta^{T} \EE_{x_0,\sigma(x_0)} \cdots \EE_{x_{T-1},\sigma(x_{T-1})} \|g^*\|_\kappa \kappa (x_T)   \\
        & \leq \beta^{T} \alpha^T \|g^*\|_\kappa \kappa(x_0)
        = (\alpha \beta)^T \|g^*\|_\kappa \kappa(x_0) \to 0
        \quad \text{as } \; T \to \infty.
    \end{align*}
    Letting $T \to \infty$, \eqref{eq:vbar} then implies that $\bar v (x_0) \geq v_\sigma(x_0)$. Since $x_0 \in \XX$ and $\sigma \in \Sigma$ are arbitrary, we have $\bar v \geq v^*$. Moreover, since $g^* = W_0 \bar v$ and there exists a $g^*$-greedy policy $\sigma^*$ by assumption, all the inequalities in \eqref{eq:vbar} holds with equality once we let $\sigma = \sigma^*$. In other words, we have $\bar{v} = v_{\sigma^*} \leq v^*$. In summary, we have shown that $\bar v = v^*$. Hence, $g^* = W_0 v^*$ and $v^* = M W_1 g^*$, and part~(a) of claim~(4) holds.
    
    Since we have shown that $v^*$ is the unique fixed point of $T$ in $\VV$, by Theorem~1 of \cite{ma2018dynamic}, the set of optimal policies is nonempty, and a feasible policy is optimal if and only if it is $v^*$-greedy. Since in addition $g^* = W_0 v^*$, parts (b) and (c) of claim~(4) hold.
\end{proof}

Next, we aim to prove Proposition~\ref{pr:suff}. For all $g \in \GG$ and $(w,z) \in \XX$, we define
\begin{equation*}
	h_g (w,z) := \max_{0 \leq s \leq w} \left\{ r(w, s) + g(z, s) \right\}
	\quad \text{and}
\end{equation*}
\begin{equation*}
	M_g (w,z) := \left\{ s \in [0, w]: h_g(w,z) = r(w,s) + g(z, s) \right\}.
\end{equation*}

The following result is helpful in applications for verifying $S \GG \subset \GG$.

\begin{lemma}
	\label{lm:well_def}
	For all $g \in \GG$, $h_g$ and $M_g$ satisfy the following properties: 
	\begin{enumerate}
		\item $h_g$ is well defined and increasing in $w$,
		\item $h_g$ is continuous on $(0, \infty) \times \mathsf{Z}$,
		\item $h_g$ is continuous on $\XX$ if $r$ is bounded below, and
		\item $M_g$ is nonempty, compact-valued, and upper hemicontinuous.
	\end{enumerate}
\end{lemma}

\begin{proof}
	Fix $g \in \GG$. Since $g$ is bounded below, 
	$h_g(0,z) = r(0,0) + g(0,z) \in \RR \cup \{ -\infty\}$ and $h_g$ is well defined at $w = 0$. Now consider $w>0$. Let $\DD_0$ be the interior of $\DD$. By assumption, either 
	\begin{enumerate}
		\item[(i)] $r$ is continuous on $\DD_0$ and $\lim_{s \to w} r(w,s) = -\infty$ for some $w \in \RR_+$, or
		\item[(ii)] $r$ is continuous and bounded below.
	\end{enumerate}
	Each scenario, since $g$ is continuous, the maximum in the definition of $h_g$ can be attained at some $s \in [0, w]$. Hence, $h_g$ is well defined for all $w>0$. Regarding monotonicity, let $w_1, w_2 \in \RR_+$ with $w_1 < w_2$. By the monotonicity of $r$, we have
	\begin{equation*}
		h_g (w_1, z) \leq \max_{s \in [0, w_1]} \{ r(w_2,s) + g(s,z) \}
		\leq \max_{s \in [0, w_2]} \{ r(w_2,s) + g(s,z) \} = h_g (w_2,z).
	\end{equation*}
	Hence, claim (a) holds. Claims (b)--(d) follow from Berge's theorem of maximum (adjusted to accommodate possibly negative infinity valued objective functions).
\end{proof}

\begin{proof}[Proof of Proposition~\ref{pr:suff}]
	$\ell$ is bounded below since, by the monotonicity of $f$ and $r$, 
	\begin{equation*}
		\ell(x,a) = \EE_{z, s} r(w',0)
		\geq \EE_{z} r[f(0,\eta'), 0] = \underline{\ell} (z),
	\end{equation*}
	which is bounded below by assumption.
	Moreover, it is obvious that $\GG$ is a closed subset of $\gG$. Existence of $g$-greedy policies for $g$ in $\GG$ has been verified by Lemma~\ref{lm:well_def}. It remains to show that $S \GG \subset \GG$. For fixed $g \in \GG$, Theorem~\ref{t:cs} implies that $Sg \in \gG$. To see that $Sg$ is increasing in its last argument and continuous, note that by Lemma~\ref{lm:well_def}, $h_g$
	is continuous on $\DD_0$ and increasing in $w'$. For all $s_1, s_2 \in \AA$ with $s_1 \leq s_2$, the monotonicity of $f$ implies that
	\begin{align*}
		Sg(z, s_1) &= \beta \EE_{z, s_1} h_g (w',z') = \beta \EE_{z} h_g (f(s_1,\eta'), z')    \\
		& \leq \beta \EE_{z} h_g (f(s_2,\eta'), z') = \beta \EE_{z, s_2} h_g (w',z') = S g(z, s_2).  
	\end{align*}
	Hence, $Sg$ is increasing in its last argument.
	In addition, the definition of $\GG$ and the monotonicity of $r$ and $f$ implies that
	\begin{equation*}
		r(f(0,\eta'), 0) - \alpha_1 \leq h_g (w',z') \leq \alpha_2 \, \kappa (w',z')
		\quad \text{for some } \alpha_1, \alpha_2 \in \RR_+.
	\end{equation*}
	Since $\kappa_e$ and $\underline{\ell}$ are continuous and the stochastic kernel $P$ is Feller, Fatou's lemma implies that $Sg (z,s) = \beta \EE_{z,s} h_g(w',z')$ is continuous.
\end{proof}

\bibliographystyle{ecta}

\bibliography{ubdp5}

\end{document}